\documentclass[a4paper,UKenglish]{lipics-v2018}

\usepackage{microtype}
\title{Consistent sets of lines with no colorful incidence}

\author{Boris Bukh}{Carnegie Mellon University, Department of Mathematical Sciences\\ {Pittsburgh, PA 15213, USA.}}{}{}{Supported in part by Sloan Research Fellowship and by U.S.\ taxpayers through NSF CAREER grant DMS-1555149. Part of the work was done during the visit to Universit\'e Paris-Est Marne-la-Vall\'ee supported by LabEx B\'ezout (ANR-10-LABX-58).}
\author{Xavier Goaoc}{Université Paris-Est, LIGM\\{UMR 8049, CNRS, ENPC, ESIEE, UPEM, F-77454, Marne-la-Vallée, France.}}{}{}{Supported by the Institut Universitaire de France.}
\author{Alfredo Hubard}{Université Paris-Est, LIGM\\{UMR 8049, CNRS, ENPC, ESIEE, UPEM, F-77454, Marne-la-Vallée, France.}}{}{}{}
\author{Matthew Trager}{Inria and \'Ecole Normale Supérieure, CNRS, PSL Research University.{}}{}{}{Supported in part by the ERC grant VideoWorld and the Institut Universitaire de France}

\authorrunning{B. Bukh, X. Goaoc, A. Hubard and M. Trager}
\Copyright{Boris Bukh, Xavier Goaoc, Alfredo Hubard and Matthew Trager}

\EventEditors{Bettina Speckmann and Csaba D. T{\'o}th}
\EventNoEds{2}
\EventLongTitle{34th International Symposium on Computational Geometry (SoCG 2018)}
\EventShortTitle{SoCG 2018}
\EventAcronym{SoCG}
\EventYear{2018}
\EventDate{June 11--14, 2018}
\EventLocation{Budapest, Hungary}
\EventLogo{socg-logo} 
\SeriesVolume{99}
\ArticleNo{} %
\nolinenumbers 
\hideLIPIcs  

\usepackage{bbold}

\theoremstyle{plain}

\newtheorem{problem}[theorem]{Problem}
\newtheorem{proposition}[theorem]{Proposition}

\newcommand\R{\ensuremath{\mathbb{R}}}
\newcommand\N{\ensuremath{\mathbb{N}}}

\newcommand\Z{\ensuremath{\mathbb{Z}}}
\newcommand\F{\ensuremath{\mathcal{F}}}
\newcommand\Li{\ensuremath{\mathcal{L}}}

\newcommand\V{\ensuremath{\mathbb{V}}}

\newcommand\1{\ensuremath{\mathbb{1}}}

\renewcommand\Pr[1]{\ensuremath{\textrm{P}\left[#1\right]}}
\newcommand\Ex[1]{\ensuremath{\textrm{E}\left[#1\right]}}
\newcommand\Var[1]{\ensuremath{\textrm{Var}\left[#1\right]}}
\newcommand\Cov[1]{\ensuremath{\textrm{Cov}\left[#1\right]}}
\newcommand{\pth}[1]{\left( #1 \right)}

\newcommand*{\eqdef}{\stackrel{\text{\tiny{def}}}{=}}

\subjclass{Theory of computation $\rightarrow$ Randomness, geometry and discrete structures, Computing methodologies $\rightarrow$ Artificial intelligence $\rightarrow$ Computer vision $\rightarrow$ Computer vision tasks $\rightarrow$ Scene understanding}
\acknowledgements{We thank \'Eric Colin de Verdi\`ere and Vojta Kalusza for discussions at an early stage of this work.}

\begin{document}
\maketitle

\begin{abstract}
  We consider incidences among colored sets of lines in $\R^d$ and
  examine whether the existence of certain concurrences between lines
  of $k$ colors force the existence of at least one concurrence
  between lines of $k+1$ colors. This question is relevant for 
  problems in 3D reconstruction in computer
  vision.
\end{abstract}

\section{Introduction}

A central problem in computer vision is the reconstruction of a
three-dimensional scene from multiple photographs.  Trager et
al.~\cite[Definition 1]{trager2016consistency} defined a set of images
as consistent if they represent the same scene from different points
of view. They constructed examples (like that of
Figure~\ref{f:inconsistent}) of a set of images which is pairwise
consistent while being altogether inconsistent.
\begin{figure}[ht]
  \includegraphics[page=2]{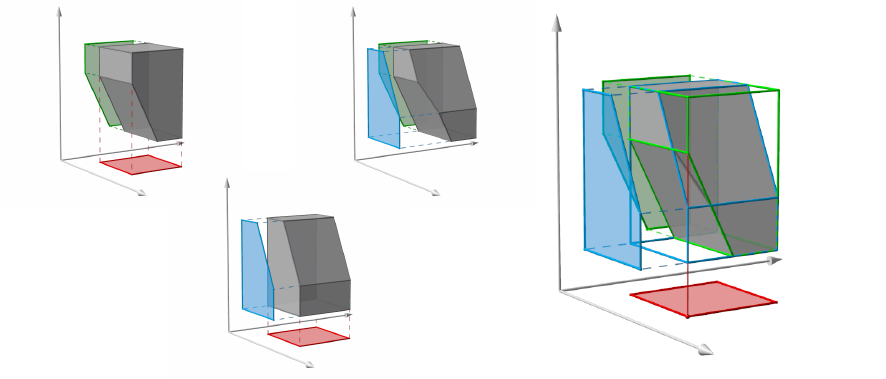}
  \caption{ Three silhouettes that are 2-consistent but not globally
    consistent for three orthogonal projections. Each of the first
    three figures shows a three-dimensional set that projects onto two
    of the three silhouettes. The fourth figure illustrates that no
    set can project simultaneously onto all three silhouettes: the
    highlighted red image point cannot be lifted in 3D, since no point
    that projects onto it belongs to the pre-images of both the blue
    and green silhouettes.\label{f:inconsistent}}
\end{figure}
They also showed~\cite[Proposition 4]{trager2016consistency} that
under a certain convexity hypothesis, images that are consistent three
at a time are globally consistent. In this paper we drop the convexity
condition and consider these affairs from the point of view of
incidence geometry.

\subparagraph{Problem statement.}

An \emph{incidence} is a set of lines that meet at a single point. Let
$\Li = \Li_1 \cup \Li_2 \cup \ldots \cup \Li_{m}$ be a set of lines of
$m$ colors in $\R^d$ (where each $\Li_i$ is a color class). Given $S
\subset \{1,2\ldots m\}$, an \emph{$S$-incidence} in $\Li$ is an
incidence between lines of every color in $S$. This paper focuses on
the following notions:

\begin{definition}
  For $1 \le k \le m$, a \emph{$k$-incidence} in $\Li$ is a
  $S$-incidence where $|S|=k$. A \emph{colorful incidence} in $\Li$ is
  an incidence that contains lines of every color.
\end{definition}

\begin{definition}
   The set $\Li$ is \emph{$k$-consistent} if for every $k$-tuple of
   colors $S\subset \{1,2\ldots m\}$, every line in $\cup_{i \in
     S}\Li_i$ belongs to an $S$-incidence. The set $\Li$ is
   \emph{consistent} if every line belongs to (at least) one colorful
   incidence.
\end{definition}

\noindent
Instead of wondering if $k$-consistency implies consistency, we aim for a more modest goal:

\begin{problem}\label{p}
  Under which conditions does the $k$-consistency assumption imply the
  existence of a $(k+1)$-incidence?
\end{problem}

\noindent
The main results of this paper are two constructions of (infinite
families of) finite sets of lines which are $k$-consistent and have no
colorful incidence. Thus, consistency does not propagate.

\begin{remark}
  Unless indicated otherwise, the set $\Li$ is assumed to be
  finite. We also assume throughout that the lines in $\Li$ are
  pairwise distinct. This has no consequence on Problem~\ref{p}:
  repeating a line in a color class is useless, and if two lines of
  distinct colors coincide, then the $k$-consistency assumption
  trivially implies that this line has a $(k+1)$-incidence.
\end{remark}

\subparagraph{Relation to photograph consistency.}

Let us explain how our initial image consistency question relates to
Problem~\ref{p}. Firstly, we ignore color or intensity information,
and treat the scene as a set of opaque objects and the images as their
projections onto certain planes. In this setting, images are
  \emph{consistent} if and only if there exists a subset $R \subset
  \R^3$ that projects into each of them. Assuming that light travels
along straight lines, the set of 3D points that are mapped to a given
image point is a ray, or more conveniently a line, in $\R^3$.
Starting with $m$ photographs, if we let $\Li_i$ denote the lines that
are pre-images of the projection on the $i$th photograph, then 
  the photographs are consistent if and only if $\cup_{i=1}^m \Li_i$
  is consistent: $R$ is the set of points of colorful incidences.

\subparagraph{Setting.}

In the basic set-up for computer vision, all lines used to project the
scene onto a given image plane pass through a ``pinhole''. We
therefore define a color class $\Li_i$ as \emph{concurrent} if it
consists of concurrent lines. We consider, however, the problem more
generally since it is possible to build other imaging systems. For
example, there are cameras that use the lines secant to two fixed skew
lines; other cameras use the lines secant to an algebraic curve
$\gamma$ and to a line intersecting $\gamma$ in $\deg \gamma -1$
points. For a discussion of the geometry of families of lines arising
in the modeling of imaging systems,
see~\cite{batog2010admissible,trager:hal-01506996} and the references
therein.

We focus in this paper on the consistency of finite sets of
lines. This restriction is technically convenient and remains relevant
to the initial motivation on continuous sets of lines. On the one
hand, our constructions for the finite problem turn out to readily
extend to infinite families of lines (see Section~\ref{s:lower}).  On
the other hand, the finite problem is already relevant to 3D
reconstruction, when one has to recover the camera parameters
(settings or position) used in the photographs. Indeed, this recovery
is typically done by identifying pixels in different images that are
likely to be the projection of the same 3D element, and using the
incidence structure of their inverse images to infer the position of
the camera; this process is called \emph{structure from
  motion}~\cite{ozyecsil2017survey}. The number of lines required to
determine the cameras is typically $5$ to $7$ per image. Although
pixels are usually matched across pairs of images, there are good
reasons for wanting to match them across more images, firstly for
robustness to noise, but also because this avoids ambiguities in the
reconstruction in the case of degenerate camera configurations (for
example, pairwise matches are never sufficient to reconstruct a scene
from images when all the camera pinholes are exactly
aligned~\cite[Chapter 15.4.2]{hartley2003multiple}).  Understanding
the consistency propagation may simplify the certification of such
matchings.

\subsection{Results}
\label{s:results}

We focus on Problem~\ref{p} for $k \ge 3$ because examples of tricolor
sets of lines that are $2$-consistent but without a colorful incidence
are relatively easy to build:

\begin{example}
  Let $(\vec{x}_0,\vec{x}_1,\vec{x}_2)$ be a basis of $\R^3$. Let $p_0,p_2\ldots,
  p_{3n}$ be a set of points where $p_0$ is arbitrary, $p_{i+1} \in
  p_i + \R \vec{x}_{(i \hspace{-0.1cm}\mod 3)}$ and $p_{3n}=p_0$. For each $i \in \{0,1,2\}$
  define $\Li_i$ to be the set of lines in coordinate direction $i$
  that are incident to points $p_j$ with $j\equiv i\pmod 3$ and $j\equiv i-1\pmod 3$. 
  If desired, we may apply a projective
  transformation that turns parallelism into concurrence.
\end{example}

\subparagraph{Constructions from higher-dimensional grids.}

We present two constructions of arbitrary large sets of lines in
$\R^d$ of $k+1$ colors that are $k$-consistent and have no colorful
incidence, for every $k \ge 3$ and $k+1 \ge d \ge 2$. Both
constructions are based on selecting subsets of lines from a regular
grid in $\R^{k+1}$. In one case, the selection is probabilistic
(Theorems~\ref{t:proba}), while in the other case it uses linear
algebra over finite vector spaces (Theorem~\ref{t:alg}). In both
constructions, every color class is concurrent. The probabilistic
argument is asymptotic and proves the existence of configurations
where every line is involved in many $k$-incidences for every choice
of $k-1$ other colors. The algebraic construction is explicit and is
minimal in the sense that removing any line breaks the
$k$-consistency.

\subparagraph{Restrictions on higher-dimensional grids.}

We then test the sharpness and potential of constructions from
higher-dimensional grids. On the one hand, we examine the number of
lines of such constructions. The algebraic selection method picks at
least $2^{k^2-k-1}$ lines of each color (we leave aside the
probabilistic selection method as its analysis is asymptotic). This
construction has the property that the lines meeting at a
$k$-incidence are not ``flat'', in the sense that they are not
contained in a $k-1$-dimensional subspace. We show, using the
polynomial method~\cite{guth2010algebraic}, that for any construction
with this property, the number of lines must be at least exponential
in $k$ (Proposition~\ref{p:flatd}). On the other hand, we examine the
possibility of designing similar constructions for models of cameras
in which the lines are not all concurrent.  We observe that when every
color class is secant to two fixed lines, lines from two color classes
cannot form a complete bipartite intersection graph
(Proposition~\ref{p:nogrid}).

\subparagraph{Small configurations.}

We also investigate small-size configurations of lines in $\R^3$ with
$4$ colors that are $3$-consistent but have no colorful incidence. The
smallest example provided by our constructions has $32$ lines per
color, which says little for applications like structure from motion,
where each color class has very few lines. Figure~\ref{f:3333} shows
two non-planar examples with $12$ lines each. We prove that they are
the only non-planar constructions with these parameters
(Theorem~\ref{t:small-i}). We also show that any configuration with
these parameters and concurrent color classes must have at least $24$
lines or be planar (Theorem~\ref{t:small-ii}).

\begin{figure}[ht]
  \begin{center}
    \includegraphics[page=1]{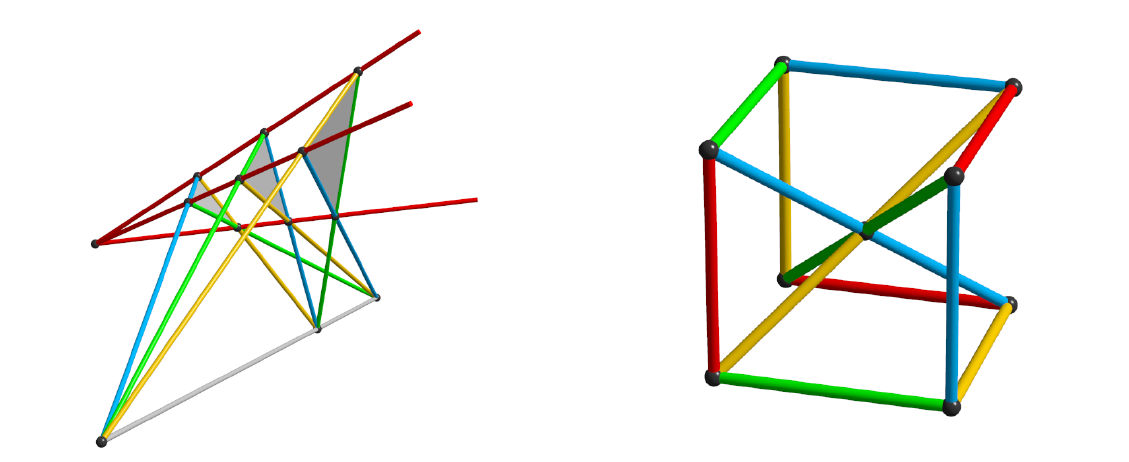}
  \end{center}
  \caption{Two non-planar examples of $12$ lines in $4$ colors that
    are $3$-consistent and have no $4$-incidence. (Left) A variation
    around Desargues' configuration. (Right) A subset of the
    $(12_416_3)$ configuration of Reye; note that triples of parallel
    lines intersect at infinity.\label{f:3333}}
\end{figure}

\subsection{Related work}

The study of consistent families of colored lines relates most
prominently to classical questions in computer vision and in discrete
geometry.

\subparagraph{In computer vision.}

The simplest and most extensively studied setting for consistency
deals with families where each color class has a single line. The
study of $n$-tuples of lines that are incident at a point (or ``point
correspondences''), is central in {\em multi-view
  geometry}~\cite{hartley2003multiple}, that is the foundation of
3D-reconstructions algorithms. In this setting, consistency propagates
trivially: $n$ lines are concurrent if and only if any three of them
are (even better: $n$ lines not all coplanar are concurrent if and
only if every pair of them is). Concurrency constraints are
traditionally expressed algebraically as polynomials in image
coordinates (see, {\em e.g.}~\cite{faugeras1995geometry}).

A more systematic study of consistency for silhouettes ({\em i.e.},
for infinite families of lines) was proposed to design reconstruction
methods based on shapes more complex than points or
lines~\cite{boyer2006using,hernandez2007silhouette}. Pairwise
consistency for silhouettes can be encoded in a ``generalized epipolar
constraint'', which can be viewed as an extension of the epipolar
constraint for points, and expresses $2$-consistency in terms of
certain simple tangency
conditions~\cite{aastrom1996generalised,trager2016consistency}. There
is no known similar characterization for $k$-consistency with
$k>2$. Consistency propagation is only known for convex silhouettes:
$3$-consistency implies consistency~\cite{trager2016consistency}.

In the dual, consistency expresses conditions for a family of planar
sets to be sections of the same 3D
object~\cite{trager2016consistency}, a question classical in geometric
tomography or stereology. We are not aware of any relevant result on
consistency in these directions.

\subparagraph{Discrete geometry.} 

As evidenced by Figure~\ref{f:3333}, our analysis of small
configurations relates to the classical configurations of Reye and
Desargues in projective geometry. Our problem and results for larger
configurations relate to various lines of research in incidence
geometry. Inspired by the Sylvester--Gallai theorem,
Erd\H{o}s~\cite{erdopurdy} asked for the largest number of collinear
$k$-tuples in a planar point set with no collinear $k+1$-tuple. The
best construction for $k=3$ come from irreducible cubic
curves\footnote{This case is closely connected with the famous
  \emph{orchard problem} recently solved in its asymptotic
  version~\cite{green2013sets}}. For higher $k$ the best construction
was given by Solymosi and Stojakov\'ic~\cite{solymosi2013many} and are
projections of higher-dimensional subsets of the regular grid
(selected, unlike ours, by taking concentric spheres). In the plane,
our problem is dual to a colorful variant of Erd\H{o}s's question.  An
intermediate between Erd\H{o}s's problem and the one treated here would
ask for the existence of a set of lines $\Li$ in which each line is
involved in many (colorless) $k$-incidences but there are no
(colorless) $k+1$-incidences. Since the Solymosi-Stojakov\'ic
construction provides $n^{2-\frac{c}{\sqrt{n}}}$ aligned $k$ tuples of
points, it is not hard to see, using a greedy deletion argument, that
this alternative problem is essentially equivalent to Erd\H{o}s's
original one.

In higher dimensions, the question of finding sets of lines with many
$k$-rich points (in the terminology of~\cite{guth2017ruled}) is
interesting even without the condition of having no $(k+1)$\nobreakdash-rich
point. Much of the recent research around this question has followed
the solution to the joint problem~\cite{guth2010algebraic} and has
been driven by algebraic considerations (see~\cite{guth2017ruled} and
the references therein). Here, we also ask for many $k$-rich points,
but our questions are driven by combinatorial considerations. Our
assumptions trade the usual density requirements (we assume linearly
many, rather than polynomially many, $k$-rich points) for structural
hypotheses in the form of conditions on the colors. On the other hand,
we can use some of the algebraic methods; the proof of
Proposition~\ref{p:flatd} is, for instance, modeled on the upper bound
on the number of joints of Guth and Katz~\cite{guth2010algebraic}.

\section{Probabilistic construction}
\label{s:proba}

In this section we prove:

\begin{theorem}\label{t:proba}
  For any $k \ge 3$, $k+1 \ge d \ge 2$, and arbitrarily large $N \in
  \N$, there exists a finite set of lines in $\R^d$ of $k+1$ colors
  that is $k$-consistent, has no $(k+1)$-incidence, and in which each
  color class consists of between $N$ and $3N$ lines, all concurrent.
\end{theorem}

\noindent
We describe our construction in $\R^{k+1}$ with color classes
consisting of parallel lines. We then apply an adequate projective
transform (to turn parallelism into concurrence) and a generic
projection to a $d$-dimensional space; both transformations preserve
incidences and therefore the properties of the
construction.

\subparagraph{Construction.}

Consider the finite subset
$[n]^{k+1}=\{1,2,...,n\}^{k+1} \subset \R^{k+1}$ of the integer grid. We make our
construction in two stages:
\begin{itemize}
\item Consider the set $\Li_i^{\#}$ of $n^k$ lines that are parallel
  to the $i$th coordinate axis and contain at least one point of our
  grid. We pick a random subset $\Li_i'$, where each line from
  $\Li_i^{\#}$ is chosen to be in $\Li_i'$ independently with
  probability $p\eqdef 2n^{-\frac2{2k-1}}$ (the value of $p$ is chosen
  with foresight).
\item We then delete from $\Li_i'$ all lines that are concurrent with
  $k$ other lines from $\cup_{j \neq i} \Li_j'$ and denote the
  resulting set $\Li_i$.
\end{itemize}
We let $\Li$ denote the colored set of lines $\Li = \Li_1 \cup \Li_2
\cup \ldots \cup \Li_{k+1}$. The second stage of the construction
ensures that $\Li$ has no $(k+1)$-incidence.\footnote{Deleting one
  line per concurrence of size $k+1$ would suffice, but deleting all
  lines as we do simplifies the analysis and suffices for our
  purpose.}  To prove Theorem~\ref{t:proba}, it thus suffices to show
that with positive probability, $\Li$ is $k$-consistent and each
$\Li_i$ has the announced size. Let us clarify that all lines
considered in the proof are in $\cup_{i=1}^{k+1} \Li_i^{\#}$ unless
stated otherwise.

\subparagraph{Consistency.}

Let us argue that $\Li$ is $k$-consistent with high probability. For a
set $I\subset [k+1]$, let
\begin{align*}
  S_I&\eqdef \{ Q \in [n]^{k+1} : \forall i \in I \text{ there is a line of }\Li_i\text{ containing }Q\},\\
  S_I'&\eqdef \{ Q \in [n]^{k+1} : \forall i \in I \text{ there is a line of }\Li_i'\text{ containing }Q\}.
\end{align*}
We say that $\ell\in \Li_i^{\#}$ is \emph{$j$-bad} (for $j\neq i$) if
$\ell$ contains no point of $S_{[k+1]\setminus \{i,j\}}$. Note that
$\Li$ is not $k$-consistent precisely when some $\ell\in \Li_i^{\#}$
is $j$-bad and $\ell$ ends up in~$\Li_i$.

Let $\ell\in \Li_i^{\#}$ and let $L\subset \Li_i^{\#}$ be any set
containing $\ell$. Let $j\neq i$.  We shall estimate $ \Pr{(\ell \in
  \Li_i) \wedge (\ell\text{ is $j$-bad}) \mid \Li'_i=L}$.  For ease of
notation, we may assume that $i=k+1$, $j=k$ and $\ell$ is the line
$\{(1,1,\dotsc,1,x) : x\in \R\}$.  Call a point $Q\in [n]^{k+1}$
\emph{regular} if $Q\notin \ell$.

The randomness in the construction comes from $(k+1)n^k$ random
choices, one for each line in $\Li_1^{\#}\cup\dotsb\cup
\Li_{k+1}^{\#}$.  We refer to these random choices as `coin flips'
since we can think of each as being a result of a toss of a (biased)
coin.

Let $\ell_{r,x}$ denote the line $\{(1,1,\dotsc,1,y,1,\dotsc,1,x):
y\in \R\}$, where $y$ is at position~$r$.  If a line $\ell' \notin
\Li_{r}^{\#}$ intersects $\ell_{r,x}$ in point
$(1,1,\dotsc,1,y,1,\dotsc,1,x)$, then all points of $\ell'$ have $y$
in the $r$th position. Note that a point
$(1,1,\dotsc,1,y,1,\dotsc,1,x)$ is regular if $y\neq 1$. A crucial
observation is that if a line $\ell'\not\in \Li_{k+1}^{\#}$ intersects
$\ell_{r,x}$ in a regular point and a line $\ell''\not\in
\Li_{k+1}^{\#}$ intersects $\ell_{r',x'}$ in a regular point and
$(r,x)\neq (r',x')$, then $\ell'$ is different from $\ell''$. This
implies that sets of coin flips on which the events of the form
\begin{center}
 ``there is a regular $Q\in \ell_{r,x}\ Q\in S'_{[k]\setminus \{r\}}$''
\end{center}
are disjoint for distinct $(r,x)$, apart from the flips associated to the lines in $\Li_{k+1}^{\#}$.

For a point $Q\in [n]^{k+1}$, let $\lambda(Q)$ be the line in $\Li_{k+1}^{\#}$ containing~$Q$.
Hence,
\begin{align*}
  &\Pr{(\ell \in \Li_{k+1}') \wedge (\ell\text{ is $k$-bad}) \mid \Li'_{k+1}=L}\\
  &=\Pr{(\ell \in \Li_{k+1}') \wedge \bigwedge_{x\in [n]} (1,1,\dotsc,1,x)\notin S_{[k-1]} \mid \Li'_{k+1}=L}\\
  &=\Pr{(\ell \in \Li_{k+1}') \wedge \bigwedge_{x\in [n]} \left(\exists r\in [k-1]\ \ell_{r,x}\notin \Li_r\right) \mid \Li'_{k+1}=L}\\
  &=\Pr{(\ell \in \Li_{k+1}') \wedge \bigwedge_{x\in [n]} \left(\exists r\in [k-1]\ \ell_{r,x}\notin \Li_r'\vee (\exists Q\in \ell_{r,x} \cap S'_{[k+1]})\right) \mid \Li'_{k+1}=L}\\
\intertext{In this last formula, the point $Q$ can be assumed to be regular because $\ell\in L$, by assumption. Now we may drop $\ell\in \Li'_{k+1}$ to obtain that the above is }
&\leq \Pr{\bigwedge_{x\in [n]} \left(\exists r\in [k-1]\ \ell_{r,x}\notin \Li_r'\vee(\exists \text{\,reg. }Q\in \ell_{r,x} \cap S'_{[k+1]})\right) \mid \Li'_{k+1}=L}\\
\intertext{Observe that if $\ell_{r,x}\in \Li_r'$ then $Q\in \ell_{r,x} \cap S'_{[k+1]}$ holds if and only if $Q\in \ell_{r,x} \cap S'_{[k]\setminus \{r\}}$ and $\lambda(Q)\in L$. By the observation above, the set of coin flips on which these latter events depend for different $x$ are disjoint, so this probability is}
&=\prod_{x\in [n]}\Pr{\exists r\in [k-1]\ \ell_{r,x}\notin \Li_r'\vee(\exists \text{\,reg. }Q\in \ell_{r,x} \cap S'_{[k]\setminus \{r\}} \wedge \lambda(Q)\in L) \mid \Li'_{k+1}=L}\\
&=\prod_{x\in [n]}\left(1-\Pr{\forall r\in [k-1]\ \ell_{r,x}\in \Li_r'\wedge(\forall \text{\,reg. }Q\in \ell_{r,x}\setminus S'_{[k]\setminus \{r\}} \vee \lambda(Q)\notin L) \mid \Li'_{k+1}=L}\right)\\
&=\prod_{x\in [n]}\left(1-\prod_{r\in [k-1]}\left(p \cdot \prod_{\substack{\text{regular } Q\in \ell_{r,x}\\\lambda(Q)\in L}}\Pr{Q\notin S'_{[k]\setminus \{r\}}}\right)\right)\\
\intertext{Call $L\subset \Li_{k+1}^{\#}$ \emph{unbiased} if for every pair $(r,x)\in [k-1]\times [n]$ the number of points $Q\in \ell_{r,x}$ such that $\lambda(Q)\in L$ is at most
$2pn$. For unbiased $L$, we obtain that the above is}
& \leq \left(1-\left(p \cdot (1-p^{k-1})^{2pn}\right)^{k-1}\right)^n
\leq \left(1-\left(\tfrac{1}{2}p \right)^{k-1}\right)^n\leq e^{-n \pth{\tfrac{1}{2}p}^{k-1}} = e^{-n^{\frac{1}{2k-1}}}\\
\end{align*}
If we pick $L$ uniformly at random, then, for every
$(r,x)\in[k-1]\times[n]$, the number of points $Q\in \ell_{r,x}$ such
that $\lambda(Q)\in L$ is a binomial random variable. With help from
Chernoff's bound, we then obtain
\begin{align*}
  & \Pr{(\ell \in \Li_{k+1}') \wedge (\ell\text{ is  $k$-bad})}\\
  &\leq \Pr{\Li_{k+1}\text{ is biased}} + \sum_{\text{unbiased }L} \Pr{\Li'_{k+1}=L} \Pr{(\ell \in \Li_{k+1}')
    \wedge (\ell\text{ is $k$-bad}) \mid \Li'_{k+1}=L}\\
  &\leq \sum_{(r,x)\in [k-1]\times[n]} e^{-(pn)^2/2n}+e^{-n^{\frac{1}{2k-1}}}=e^{-cn^{\frac{1}{2k-1}}}.
\end{align*}
By taking the union bound over all $i$, $j$ and $\ell$ we obtain that
\begin{align*}
  \Pr{\Li\text{ is not }k\text{-consistent}}&\leq\Pr{\exists i,j\ \exists \ell\in \Li_i^{\#}\ \bigl((\ell \in \Li_{i}') \wedge (\ell\text{ is $j$-bad})\bigr)}\\&\leq (k+1)^2n^ke^{-cn^{\frac{1}{2k-1}}}\leq e^{-c'n^{\frac{1}{2k-1}}}.
\end{align*}

\subparagraph{Size.}

We now analyze the probability that $\Li_1$ is large (the bound will
hold for each $\Li_i$). Let us write $\Li' = \cup_{i=1}^{k+1} \Li_i'$
and label $\ell_1, \ell_2, \ldots, \ell_{n^k}$ the lines parallel to
the $1$st coordinate axis that intersect our grid. Put $X_i =
\1_{\ell_i \in \F}$ and let $X = |\Li_1| = X_1+X_2+\ldots +
X_{n^k}$. We have
\[ \Ex{X_i} = \Pr{X_i = 1} = \Pr{\ell_i \in \Li'} \Pr{\ell_i \in \Li \mid \ell_i \in \Li'}  = p(1-p^k)^n,\]
so
\[ \Ex{X} =n^k p (1-p^k)^n = \pth{1-n^{-\frac{2k}{2k-1}}}^n n^{k-\frac2{2k-1}} \ge \pth{1-\frac1n}^n  n^{k-\frac2{2k-1}} \ge \frac14 n^{k-\frac2{2k-1}}.\]
Thus $\Ex{X}=N\in [\frac14 n^{k-\frac2{2k-1}},
  n^{k-\frac2{2k-1}}]$. We next use a concentration inequality to pass
from $\Ex{X}$ to an estimate on the probability that $X$ is large. The
second step introduces some dependency between some of the variables
$X_i$, so we use Chebychev's inequality:
\[ \Pr{|X-\Ex{X}|> \lambda \Ex{X}} \le \frac{\Var{X}}{\lambda^2 \Ex{X}^2}.\]
Taking $\lambda =1/2$, we get
\[  \Pr{|X-\Ex{X}|> \frac12 \Ex{X}} \le  \frac{\Var{X}}{N^2} \le 64 \Var{X}n^{-\pth{2k-\frac4{2k-1}}}.\]
Recall that
\[ \Var{X} = \sum_{i=1}^{n^k} \Var{X_i} + \sum_{1 \le i < j \le n^k} \Cov{X_i,X_j}.\]
Since $X_i$ takes values in $\{0,1\}$, the first sum in the right-hand
term is bounded by $n^k$. Moreover, there are $O(n^{k+1})$ pairs of
variables $X_i$ and $X_j$ with non-zero covariance, since this
requires the two lines $\ell_i$ and $\ell_j$ to belong to a common
axis-aligned $2$-plane. Again, each non-zero covariance is at most
$1$. Altogether, $\Var{X} = O(n^{k+1})$, so
\[ \Pr{|X-\Ex{X}|> \frac12 \Ex{X}} = O\pth{n^{\frac{2k+3}{2k-1}-k}}.\]
For $k \ge 3$, the probability that $X$ is in $\left[\frac18
  n^{k-\frac2{2k-1}}, \frac32 n^{k-\frac2{2k-1}}\right]$ goes to $1$ as $n$
goes to infinity.

\section{Algebraic construction}\label{s:alg}

In this section we prove:

\begin{theorem}\label{t:alg}
  For any $k \ge 3$, $k+1 \ge d \ge 2$, and arbitrarily large $N$,
  there exists a finite set of lines in $\R^d$ with $k+1$ colors that
  is $k$-consistent, has no $(k+1)$-incidence, and in which each color
  class consists of $N$ lines, all concurrent.
\end{theorem}

\noindent
As in Section~\ref{s:proba}, we describe our construction in
$\R^{k+1}$ with parallel families of lines, and obtain the desired
configuration by an adequate projective transformation and a
projection. We again consider the finite portion of the integer grid
$[n]^{k+1} \subset \R^{k+1}$ and the axis-aligned lines that
intersects it. Unlike in Section~\ref{s:proba}, we give an explicit
way to select some of these lines to achieve the desired
configuration.

\begin{figure}[ht]
  \begin{center}
    \includegraphics[page=3]{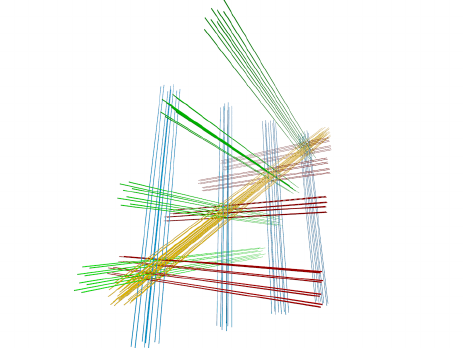}
  \end{center}
  \caption{A projection to $\R^3$ of our construction for $k=3$ and $p=2$ (reprojected to the plane). \label{f:algebraic}}
\end{figure}

\subparagraph{Construction.}

We work with axis-aligned lines that intersect in points of our
grid. Hence, identifying each line with the subset of the grid that it
contains does not affect incidences. We fix a prime number $p$ and
parameterize $[n]$ by the vector space $\V = \pth{\Z/ p \Z}^{k-1}$;
this restricts the choice of $n$ to certain prime powers, but still
allows to make it arbitrarily large. We use this parametrization to
describe the lines in our configuration as solutions of well-chosen
linear equations.

Let $v_1, v_2, \ldots, v_k \in \V$ such that $v_1+v_2+ \ldots + v_k =
0$ and any proper subset of them are linearly independent. Let $\cdot$
denote the inner product of the vector space $\V$. For $i = 1 \ldots
k$, our set $\Li_i$ consist of all the lines parallel to the $i$th
coordinates and passing through a point with parameters $(X_1, \ldots,
X_{k+1}) \in \V^{k+1}$ such that
\begin{equation}\label{e:alg1}
  v_{i-1} \cdot X_1 +  v_{i-1} \cdot X_2 + \ldots +  v_{i-1} \cdot X_{i-1} + v_i\cdot X_{i+1} + \ldots +  v_i\cdot X_{k+1} = 0.\end{equation}
(Keep in mind that each $X_i$ is a vector in $\pth{\Z/ p \Z}^{k-1}$.) We define $\Li_{k+1}$ similarly but replace Equation~\eqref{e:alg1} by
\begin{equation}\label{e:alg2}
  v_k \cdot X_1 + v_k \cdot X_2 + \ldots + v_k \cdot X_k =1.
\end{equation}

\subparagraph{No $(k+1)$-incidence.}

Any $(k+1)$-incidence is a point of the grid whose parameters $(X_1,
\ldots, X_{k+1})$ satisfy the system:
\[
\left\{
\begin{array}{lllllllllllll}
& & v_1 \cdot X_2 & + & v_1 \cdot X_3 & + & \ldots & + & v_1 \cdot X_k & + & v_1 \cdot X_{k+1} & = 0\\
v_1 \cdot X_1 & & & + & v_2 \cdot X_3 & + & \ldots & + & v_2 \cdot X_k & + & v_2 \cdot X_{k+1} & = 0\\
v_2 \cdot X_1 & + & v_2 \cdot X_2 & & & + & \ldots & + & v_3 \cdot X_k & + & v_3 \cdot X_{k+1} & = 0\\
v_3 \cdot X_1 & + & v_3 \cdot X_2 & + &  v_3 \cdot X_3 & + & \ldots & + & v_4 \cdot X_k & + & v_4 \cdot X_{k+1} & = 0\\
\ldots &&&&&&&&&&&&\\
v_{k-1} \cdot X_1 & + & v_{k-1} \cdot X_2 & + &  v_{k-1} \cdot X_3 & + & \ldots & &  & + & v_k \cdot X_{k+1} & = 0\\
v_{k} \cdot X_1 & + & v_{k} \cdot X_2 & + &  v_{k} \cdot X_3 & + & \ldots & + & v_k \cdot X_k  & & & = 1\\
\end{array}
\right.\]
Summing all these conditions yields
\[ \pth{\sum_{i=1}^k v_i} \cdot \pth{\sum_{i=1}^{k+1} X_i} = 1,\]
which contradicts $v_1+v_2+ \ldots + v_k = 0$. So there is no $(k+1)$-incidence.

\subparagraph{$k$-consistency.}

Fix a line $\ell \in \Li_1$. It corresponds to some solution
$(X_2^*,\ldots, X_{k+1}^*)$ of Equation~\eqref{e:alg1}. The grid
points on $\ell$ are precisely the points of the form $(X_1, X_2^*,
X_3^*, \ldots, X_{k+1}^*)$ and are parametrized by $X_1$. Each
equation in the system above reduces to $v_j \cdot X_1 = c_j$, where
$c_j$ is some constant vector (computed from the $v_j$'s and the
$X_j^*$'s). Since $X_1 \in \pth{\Z/ p \Z}^{k-1}$ and any $k-1$ of the
$v_j$ are linearly independent, any choice of $k-1$ equations has a
solution. This means that for any $i$, the line $\ell$ is concurrent
with lines from all $\Li_j$ with $j \in [k+1] \setminus \{i\}$. The
same goes with the lines of $\Li_2, \ldots, \Li_{k+1}$ so the
configuration is consistent.

\subparagraph{Size.}

In this construction, the size of $\Li_i$ is the number of $(X_1,
\ldots, X_{i-1},X_{i+1}, \ldots, X_{k+1})$ in $\V^{k}$ satisfying
Equation~\eqref{e:alg1} -- or~\eqref{e:alg2} if $i=k+1$. Hence
$|\Li_i| = p^{k^2-k-1}$ for every $i$. The smallest 
configuration built in this way thus has $2^{k^2-k-1}$ lines per set
(which is $32$ for $k=3$); refer to Figure~\ref{f:algebraic}.

\section{More on grid-like examples} 
\label{s:lower}

Both Theorems~\ref{t:proba} and~\ref{t:alg} construct examples as
projections of subsets of a regular grid in higher dimension. We
discuss here the properties of such constructions.

\subparagraph{Number of lines.} 

Consider a colored set of lines $\Li = \Li_1 \cup \Li_2 \cup \ldots
\cup \Li_{k+1}$ in $\R^d$. We say that a $t$-incidence of $\Li$ is
\emph{flat} if the lines meeting there are contained in an affine
subspace of dimension at most $\min(d,t)-1$. In any grid-like
construction such as those in Theorems~\ref{t:proba} and~\ref{t:alg}, every
$k$-incidence is non-flat.

\begin{proposition}\label{p:flatd}
  Let $\Li = \Li_1 \cup \Li_2 \cup \ldots \cup \Li_{m}$ be a
  $k$-consistent colored set of lines in $\R^k$ with no
  $(k+1)$-incidence. If no $k$-incidence of $\Li$ is flat, then
  \[ \sum_{i=1}^m |\Li_i| \ge \frac{\binom{\binom{m-1}{k-1} +k-1}{k}}{\binom{m-1}{k-1}}.\]
\end{proposition}

\begin{proof}
  The proof essentially follows the argument of Guth and
  Katz~\cite{guth2010algebraic} for bounding the number of joint among
  $n$ lines; the main difference is that the consistency assumption
  makes their initial pruning step unnecessary. We spell it out for
  completeness.

  Let $P$ denote a set of concurrency centers witnessing all the
  $k$-incidences required by the $k$-consistency. We choose $P$ of
  minimum size, so that
  \begin{equation}\label{eq:p}
    |P| \le \pth{\sum_{i=1}^m |\Li_i|} \binom{m-1}{k-1}.
    \end{equation}
  Let $f(x_1,\ldots, x_k)$ be a nontrivial (not necessarily
  homogeneous) polynomial that vanishes at every point of $P$ and has
  minimal total degree. We claim that
  \begin{equation}\label{eq:f}
    \binom{m-1}{k-1} \le \deg f \quad \hbox{and} \quad \binom{\deg f +k-1}{k} \le |P|.
  \end{equation}
  Note that inequalities~\eqref{eq:p} and~\eqref{eq:f} together imply
  the statement.
  
  It remains to prove inequalities~\eqref{eq:f}. The second inequality
  follows from the minimality of $\deg f$; if it were false, we would be able to find a non-zero polynomial
  of degree $\deg f-1$ vanishing on $P$ by solving for its coefficients. 
  For the first inequality, we argue
  \emph{ab absurdum}. Assume that it fails. Then every line in $\cup_i
  \Li_i$ intersects $\{f=0\}$ in strictly more points than the degree of
  $f$. This implies that every line in $\cup_i \Li_i$ is contained in
  $\{f=0\}$. Since no incidence is flat, this in turn implies that 
$\nabla f=(\frac{\partial f}{\partial x_1},\dotsc,\frac{\partial f}{\partial x_d})$
vanishes at every point of $P$. So, every polynomial of the form $\frac{\partial f}{\partial
    x_j}$ vanishes on all of $P$; all these polynomials have total
  degree strictly smaller than $f$. Since $f$ is non-constant, then at least one of these polynomials is
  nontrivial.  This contradicts the minimality of the total degree of
  $f$.
\end{proof}

For $m=k+1$, the bound of Proposition~\ref{p:flatd} is $\frac1k
\binom{2k}k$, so the number of lines required grows exponentially with
$k$.

\subparagraph{Non-concurrent colors.} 

Theorems~\ref{t:proba} and~\ref{t:alg} both use a grid in $\R^{k+1}$
to start with $k+1$ color classes, each of size $n^k$, where every
line is involved in $n$ colorful incidences. Recall that in this
setup, every color class is concurrent (it consists of parallel
lines). This is in fact important, perhaps essential. To see this,
note that any two of our starting color classes contain arbitrarily
large subsets whose intersection graph is dense. This is impossible,
generically, if we try to work with families of lines that are secant
to two skew lines in $\R^3$. \footnote{This choice is motivated by the
  design of {\em two-slit
    cameras}~\cite{trager:hal-01506996,batog2010admissible}.}

\begin{proposition}\label{p:nogrid}
  For $i=1,2$, let $\Gamma_i$ denote the set of lines secant to two
  fixed lines $s_i$ and $s_i'$ in $\R^3$. Let $A$ and $B$ be two sets
  of $n$ lines from $\Gamma_1$ and $\Gamma_2$, respectively. If the
  lines $s_1$, $s_1'$, $s_2$ and $s_2'$ are in generic position then
  the intersection graph of $A$ and $B$ has $O\pth{n^{4/3}}$ edges.
\end{proposition}
\begin{proof}
  First, note that the intersection graph of $A$ and $B$ is
  semi-algebraic: parametrizing $\Gamma_i$ by $s_i \times s_i' \simeq
  \R^2$ makes the incidence an algebraic relation, as can be deduced
  from the bilinearity of incidence in Pl\"ucker coordinates.

  Next, remark that if this graph contains a complete bipartite
  subgraph $K_{3,3}$, then the lines $\{s_1,s_1',s_2,s_2'\}$ are in a
  special position. Indeed, in the generic case, these two triples of
  lines come from the two families of rulings of a quadric
  surface~\cite[$\mathsection 10$]{veblen1918projective}; the lines
  $s_1$, $s_1'$, $s_2$ and $s_2'$ are also rulings of that quadric, so
  \emph{both} $s_1$ and $s_1'$ intersect \emph{both} $s_2$ and
  $s_2'$. In the non-generic cases, the six lines must be coplanar
  with $s_1$ and $s_2$.

  Now, we apply the semi-algebraic version of the
  K\H{o}v\'ari--S\'os--Tur\'an theorem~\cite{fox2014semi}, and obtain that
  the number of edges of our graph is $O\pth{n^{4/3}}$.
\end{proof}

\noindent
 We see the previous result as indication that a straight forward adaptation of our probabilistic
construction to the case of two-slits is impossible. We would like to improve the bound in Proposition~\ref{p:nogrid} from
$O\pth{n^{4/3}}$ to $O(n)$.

\begin{remark}
  Note that the genericity assumption in Proposition~\ref{p:nogrid} is on
  the sets $\Gamma_i$, not on their subsets. The analogue for
  concurrent sets of lines would be to require that the centers of
  concurrence are in generic position; this clearly does not prevent
  finding arbitrarily large subsets with dense intersection graphs.
\end{remark}

\subparagraph{Extension to continuous sets of lines.}

The constructions of Theorems~\ref{t:proba} and~\ref{t:alg} can be
turned into continuous families of lines as follows.

First, we follow either construction up to the point where we have a
family $\Li$ of lines of $k+1$ colors in $\R^{k+1}$ that is
$k$-consistent, without colorful incidence, and where each color class
is parallel. Consider a parameter $\epsilon>0$, to be fixed later. For
every $i$, we build a set $\Li_i(\epsilon)$ by considering every line
$\ell \in \Li_i$ in turn, and adding to $\Li_i$ every line $\ell'$
parallel to $\ell$ such that the distance between $\ell$ and $\ell'$
is most $\epsilon$. Note that for $\epsilon < 1/2$ the family
$\Li(\epsilon)$ is $k$-consistent and without colorful incidence.

Now, consider a generic projection $f\colon \R^{k+1} \to \R^d$ for the
desired $d$. For any $\epsilon>0$ the family $\Li(\epsilon)$ is
$k$-consistent. We observe that for $\epsilon>0$ small enough, it also
remains without colorful incidence. Let $\tau$ denote the minimum
distance, in the projection, between a $k$-incidence and a line (of
any color) not involved in that incidence. Every $k$-incidence in
$\Li$ gives rise, in $\Li(\epsilon)$, to $k$ tubes that intersect
in a bounded convex set $B$ of size $O(\epsilon)$.
Choosing
$\epsilon>0$ such that the diameter of $f(B)$ is less than $\tau/2$
ensures that the corresponding family $f(\Li(\epsilon))$ has no
colorful incidence.

For a given family of colored lines $\Li$ define the set $P_S$ to be
the set of points incident to at least one line of each of the color
classes in $S$; see Figure~\ref{f:inconsistent}. 
Notice that in our examples, for each set $S$ of $k$
colors the set $P_S$ 
is highly disconnected. As mentioned in the introduction, Trager et
al.~\cite{trager2016consistency} showed that if a family of sets of
lines is $3$-consistent and for each $S$ of size $3$, the set $P_S$ is
convex, then the whole family is consistent. An interesting open
question is whether an analogue theorem holds if instead of convexity,
we assume that for every set $S$ of size $k$, the set $P_S$ is
sufficiently connected.

\section{Constructions with few lines}

The configurations constructed in Sections~\ref{s:proba}
and~\ref{s:alg} have at least $32$ lines per color. This is
considerably larger than the sets of lines involved in some of the
questions around consistency that arise in computer vision. In the
example of structure-from-motion mentioned in introduction, when the
camera is central, every color class has only $5$ to $7$ lines. It
turns out that for sufficiently small configurations, $k$-consistency
does imply some colorful incidences:

\begin{lemma}\label{p:verysmall}
  Any $3$-consistent colored set of lines $\Li = \Li_1 \cup \Li_2 \cup
  \Li_3 \cup \Li_4$ in $\R^d$ with $|\Li_1| = |\Li_2| = |\Li_3| =
  |\Li_4|=2$ contains a colorful incidence.
\end{lemma}
\begin{proof}
  Let us prove the case where $d=2$; the general case follows by
  projecting onto a generic $2$-plane. Let $P_i$ denote the dual of
  $\Li_i$, and let $P = P_1 \cup P_2 \cup P_3 \cup P_4$. Assume, by
  contradiction, that $\Li$ contains no colorful incidence,
  \emph{i.e.} that no line intersects every $P_i$. Let $P' = P_2 \cup
  P_3 \cup P_4$ and let us apply a projective transform to map the
  points of $P_1$ to the horizontal and vertical directions,
  respectively. We call a line that contains a point of each of $P_2$,
  $P_3$ and $P_4$ a \emph{rainbow line}.

  Since $\Li$ is $3$-consistent, for any point $x \in P$ and any
  choice of $2$ other colors, there is a line through $x$ that
  contains a point of each of these colors. Since $\Li$ has no
  colorful incidence, there must exist three horizontal lines and
  three vertical lines that intersect~$P'$, and each must contain
  exactly two points of $P'$ of distinct colors. Moreover, no rainbow
  line can be horizontal or vertical. But this implies that out of the
  $9$ intersections between horizontal and vertical lines, only $5$
  (the corners and the center) can be on a rainbow line. This
  contradicts $|P'|=6$.
\end{proof}

\noindent
We prove here a slightly stronger lower bound:

\begin{theorem}\label{t:small-ii}
  Let $\Li = \Li_1 \cup \Li_2 \cup \Li_3 \cup \Li_4$ be a
  $3$-consistent colored set of lines in $\R^3$ with no colorful
  incidence and concurrent colors. If $|\Li| < 24$, then
  $\Li$ is contained in a $2$-plane.
\end{theorem}

\noindent
We do not know whether the constant $24$ is best possible. 

\subparagraph{Classification.}

We also provide a characterization of $3$-consistent, $4$-colored sets
of lines in $\R^3$ with no colorful incidence and $3$ lines per
color. Forgetting for a moment about colors, any such configuration
must consist of $12$ lines and $12$ points, every point on $3$ lines
and every line through $3$ points; in the classical tabulation of
projective configurations, they are called \emph{$(12_3)$
  configurations}. It turns out that there are $229$ possible
incidence structures meeting this description, and that every single
one of them is realizable in $\R^3$~\cite{gropp1997configurations}. To
analyze what happens when we add back the colors and the consistency
assumption, we consider two special $(12_3)$ configurations:

\begin{itemize}
\item A \emph{Reye-type configuration}\footnote{The $(12_416_3)$
  configuration of Reye consists of $12$ points and $16$ lines in
  $\R^3$ such that every point is on $4$ lines and every line contains
  $3$ points; its realizations are projectively equivalent to the $16$
  lines supporting the $12$ edges and four long diagonals of a cube,
  together with that cube's vertices and center and the $3$ points at
  infinity in the directions of its edges} is a configuration obtained
  by selecting $12$ out of the $16$ lines supporting the $12$ edges
  and four long diagonals of a cube, in a way that produces a $(12)_3$
  configuration.

\item A \emph{Desargues-type} configuration is defined from six planes
  $\Pi_1,\Pi_2, \ldots, \Pi_6$ in $\R^3$ where (i) each of
  $\{\Pi_1,\Pi_2, \ldots, \Pi_5\}$ and $\{\Pi_1,\Pi_2, \Pi_3,\Pi_6\}$
  is in general position, and (ii) $\Pi_4,\Pi_5$ and $\Pi_6$ intersect
  in a line. The configuration consists of all lines that are
  contained in exactly two planes.
\end{itemize}

\noindent
Here is our classification:

\begin{theorem}\label{t:small-i}
  Let $\Li = \Li_1 \cup \Li_2 \cup \Li_3 \cup \Li_4$ be a
  $3$-consistent colored set of lines in $\R^3$ with no colorful
  incidence.  If every color class has size $3$, and $\Li$ is not
  contained in a $2$-plane, then it is a Desargues-type or a
  Reye-type configuration colored as in Figure~\ref{f:3333}.
\end{theorem}

\subsection{Classification: Proof of Theorem~\ref{t:small-i}}
\label{s:class3}

Let $\Li = \Li_1 \cup \Li_2 \cup \Li_3 \cup \Li_4$ be a $3$-consistent
colored set of lines in $\R^3$ with no colorful incidence and that is
not coplanar. We start with a few simple observations:

\begin{enumerate}
\item Every line $\ell$ in a color class $X$ intersects the lines from
  any other color class $Y$ in at least two distinct points. Indeed,
  $\ell$ needs to intersect the lines of $Y$ to satisfy the
  $3$-consistency condition, and if it does so in a single point, then
  that point is a colorful incidence.
  
\item If all lines of a color class are coplanar, then the
  configuration is coplanar. This follows readily from the previous
  observation.

\item Given two lines $\ell, \ell'$ and a point $x$ in $\R^3$, if
  $\ell$ and $\ell'$ are skew and $x \notin \ell \cup \ell'$, then
  there is a unique line through $x$ that intersects $\ell$ and
  $\ell'$.
\end{enumerate}

\subsubsection{Geometric analysis}

We first restrict the geometry of a color class of size three. We
denote the lines of $\Li_1$ by $a_1, a_2, a_3$ ($b_i$ for $\Li_2$,
$c_i$ for $\Li_3$ and $d_i$ for $\Li_4$). We write $<a,b>$ for the
plane spanned by lines $a$ and $b$ and $b$, $abc$ for a common
incidence of $a$, $b$ and $c$, and $p_{abc}$ for the point of
concurrence of $a$, $b$ and $c$.

\begin{lemma}\label{l:3}
  If $\Li_1 = \{a_1,a_2,a_3\}$ is a color class of size three, then

  \begin{enumerate}
  \item the lines in $\Li_1$ are pairwise skew, or
  \item the lines in $\Li_1$ are concurrent and not coplanar, or
  \item $a_1$ intersects both $a_2$ and $a_3$ and every line in $\Li_2 \cup \Li_3 \cup \Li_4$ is contained in $<a_1,a_2>$ or in $<a_1,a_3>$.
  \end{enumerate}
  Moreover, every line in $\Li_2 \cup \Li_3 \cup \Li_4$ intersects
  the lines of $\Li_1$ in exactly two points.
\end{lemma}
\begin{proof}
  Assume that we are not in case~(1) and that, wlog, the lines $a_1$
  and $a_2$ are coplanar. We first claim that there is a concurrence
  $bcd$ of lines from $\Li_2$, $\Li_3$ and $\Li_4$ such that $p_{bcd}
  \notin <a_1, a_2>$. Indeed, as the configuration is not planar, we
  can assume that $b$ is not contained in $<a_1, a_2>$. The line $b$
  intersects two lines of $\Li_1$, it must intersect $a_1$ or $a_2$;
  this point cannot be $p_{bcd}$ as it would make a colorful
  incidence, so $b \not\subset <a_1, a_2>$ forces $p_{bcd} \notin
  <a_1, a_2>$.
  
  Now, $a_3$ must be coplanar with $a_1$ or $a_2$. Indeed, each of
  $b$, $c$ and $d$ has at most one point of intersection with $a_1
  \cup a_2$, and therefore intersects $a_3$. This implies that $b$,
  $c$ and $d$ are coplanar, namely in the plane spanned by $p_{bcd}$
  and $a_3$. This implies that the intersections of $b$, $c$ and $d$
  with $<a_1, a_2>$ are three distinct, aligned points. This forces
  the intersection between the plane $p_{bcd} \vee a_3$ and $<a_1,
  a_2>$ to be one of the lines $a_1$ or $a_2$.

  Let us say, wlog, that $a_3$ is coplanar with $a_1$. If $a_3$ is
  also coplanar with $a_2$ then all three lines are either coplanar or
  concurrent. The former would imply that the entire configuration is
  coplanar, so it must be the latter; this corresponds to case~(2).

  The remaining case is when $a_3$ is coplanar with $a_1$ but not with
  $a_2$. Consider a line $x$ in $\Li_2 \cup \Li_3 \cup \Li_4$ and let
  $bcd$ be a concurrence among these three colors that involves $x$. As
  observed above, either $p_{bcd} \in <a_1, a_2>$ and $x$ is contained
  in $<a_1, a_2>$, or $p_{bcd} \notin <a_1, a_2>$ and $x$ is contained
  in $<a_1, a_3>$. We are thus in case~(3).

  Finally, let us consider a line $x \in \Li_2 \cup \Li_3 \cup \Li_4$
  and count its intersections with the lines of $\Li_1$. We already
  observed that this number is at least two. In cases~(2) and~(3) it
  is straightforward that it is also at most two. So assume that we
  are in case~(1), where the lines of $\Li_1$ are pairwise
  skew. Assume that there exists a line $b \in \Li_2$ that intersects
  the lines from $\Li_1$ in three points. Let $bcd$ be an incidence
  with lines of $\Li_3$ and $\Li_4$. The point $p_{bcd}$ cannot belong
  to a line in $\Li_1$, so there is a unique line from $p_{bcd}$ that
  intersects any two given lines in $\Li_1$. Since $c \neq b$, $c$
  intersects at most one line of $\Li_1$. This contradicts one of our
  initial observations.
\end{proof}

Remark that Lemma~\ref{l:3} only assumed that one color class has size
$3$. When all color classes have size three, we can eliminate case~(3):

\begin{lemma}\label{l:3333}
  If every color class has size $3$, then every color class consists of
  lines that are concurrent and not coplanar, or pairwise skew.
\end{lemma}
\begin{proof}
  The task is to show that case~(3) in Lemma~\ref{l:3} cannot
  occur. We argue by contradiction, and assume that $a_1$ intersects
  both $a_2$ and $a_3$ and that every line in $\Li_2 \cup \Li_3 \cup
  \Li_4$ is contained in $\Pi_2 = <a_1,a_2>$ or in $\Pi_3 =
  <a_1,a_3>$.

  Wlog we can assume that $\Pi_2$ contains exactly one line of $\Li_2$
  and one line of $\Li_3$. Indeed, each line $a_i$ must intersect
  lines from each other color class in at least two points. Thus, each
  of $\Pi_1$ and $\Pi_2$ must contain at least a line from each color
  class. Let $b_1$ and $c_1$ denote the lines of $\Li_2$ and $\Li_3$
  contained in $\Pi_2$.

  The line $a_2$ must intersect lines from each of $\Li_2$ and $\Li_3$
  in two distinct points. The only point in which $a_2$ can meet lines
  from $\Pi_3$ is $a_1 \cap a_2$. Thus, two lines $b_2 \in \Li_2$ and
  $c_2 \in \Li_3$, both contained in $\Pi_3$, pass through $a_1 \cap
  a_2$, and neither $b_1$ nor $c_1$ goes through that point.

  The line $b_1$ must intersect lines from $\Li_3$ in two distinct
  points, so one of them has to be $a_1 \cap b_1$ and the other $b_1
  \cap c_1$. In particular, $b_1$ and $c_1$ do not meet on $a_1$. 

  Now, the remaining lines $b_3 \in \Li_2$ and $c_3 \in \Li_3$ meet,
  respectively, $c_1$ and $b_1$ on $a_1$. Again, this follows from the
  fact that $b_1$ (resp. $c_1$) must intersect lines from $\Li_3$
  (resp. $\Li_2$) in two distinct points.

  We now have our contradiction: on the line $a_1$, every intersection
  with a line of $\Li_2$ is also on a line of $\Li_3$. The concurrence
  of $a_1$ with lines of $\Li_2$ and $\Li_4$ is therefore colorful.
\end{proof}

From now on, let us assume that $|\Li_1| = |\Li_2| = |\Li_3| = |\Li_4|
= 3$. Lemmas~\ref{l:3} and~\ref{l:3333} force the incidence structure
to be quite regular:

\begin{lemma}\label{l:two-inc}
  Let $\Li_i$ and $\Li_j$ be two color classes. Every line of $\Li_i$
  intersects exactly two lines of $\Li_j$. For any two lines of
  $\Li_j$, there is exactly one line of $\Li_i$ that intersects them
  both.
\end{lemma}
\begin{proof}
  Let $\ell$ be a line of $\Li_i$. We know that $\ell$ intersects the
  lines of $\Li_j$ in exactly two points. Whether $\Li_j$ is pairwise
  skew or concurrent and non-coplanar, exactly one line of $\Li_j$
  must pass through each of these points.

  Now, consider the graph on $\Li_i \cup \Li_j$ where there is an edge
  between any two lines of different colors that intersect. This graph
  is bipartite, has three vertices per class, and every vertex has
  degree two (by the first statement). It must be a $6$-cycle, and the
  second statement follows.
\end{proof}

\begin{corollary}
  At most one class consists of concurrent lines.
\end{corollary}
\begin{proof}
  Assume that $\Li_1$ and $\Li_2$ both consist of concurrent
  lines. For any $a_i, a_j \in \Li_1$, there is a line of $\Li_2$ that
  is coplanar with $a_i$ and $a_j$ (since it intersects them in two
  points). The center of concurrence of $\Li_2$ therefore lies on the
  intersection of the three planes spanned by pairs of lines of
  $\Li_1$. That intersection is the center of concurrence of
  $\Li_1$. This implies that any line of $\Li_2$ can only intersect
  the lines of $\Li_1$ in a single point, a contradiction.
\end{proof}

\subsubsection{Characterization of the \texorpdfstring{$3$-incidences}{3-incidences}}

We now focus on the combinatorial incidence structure defined by
$\Li$. We represent a combinatorial configuration using a set of cubic
monomials in $a_1,a_2,a_3$, $b_1,b_2,b_3$, $c_1,c_2,c_3$,
$d_1,d_2,d_3$. Here each variable represents a line in $\mathcal L$,
and a monomial $xyz$ means that the lines $x,y,z$ are concurrent. For
convenience, we also write $A$ for $\Li_1$, $B$ for $\Li_2$, etc.

Let $\mathcal C$ be an incidence structure for $\Li$ (that is, a set
of monomials). If we fix a color, say $A$, we can construct a graph
$\mathcal C_A$ with colored nodes and labeled edges, defined as
follows. The nodes of $\mathcal C_A$ are incidence points of $\mathcal
C$ that involve lines in $A$. Finally, each edge is uniquely
associated with a line from $B, C, D$: from Lemma~\ref{l:two-inc}, we
see that each line from $B, C, D$ has exactly {\em two} incidences
with lines in $A$, so it can be used to define an edge joining two
vertices of $\mathcal C_A$. From properties of $\mathcal C$ (and using
again Lemma~\ref{l:two-inc}) we deduce the following:

  \begin{itemize}
  \item Every node of $\mathcal C_A$ has degree two.
  \item Any edge of $\mathcal C_A$ joins vertices on different lines
    of $A$, and has neighbors of different color.
  \item For $s \neq 1$ and any $i \neq j$, there there exists a unique
    edge of class $\Li_s$ joining vertices on $a_i$ and $a_j$.
  \end{itemize}

  \noindent
  From these observations, we see that edges in $\mathcal C_A$ form
  (one or several) cycles, and that any three consecutive edges have
  different colors.  In particular, the colors alternate cyclically,
  so every cycle has length a multiple of three. An elementary
  case-analysis (\emph{e.g.} distinguishing whether $\mathcal C_A$
  contains two consecutive edges that join $a_1$ and $a_2$) reveals
  that, up to relabeling, $\mathcal C_A$ must be one of two graphs:

  \begin{center}
    \includegraphics[page=6]{figures}
  \end{center}    

  \noindent
  We claim that the only possible incidence structures are isomorphic
  to:
  \[
  (I): \begin{array}{cccc}
    a_1b_2c_3 & a_1b_3d_2 & a_1c_2d_3 & b_1c_3d_2\\
    a_2b_3c_1 & a_2b_1d_3 & a_2c_3d_1 & b_2c_1d_3\\
    a_3b_1c_2 & a_3b_2d_1 & a_3c_1d_2 & b_3c_2d_1\\
  \end{array} \quad \hbox{and} \quad 
  (II): \begin{array}{cccc}
    a_1b_2c_3 & a_1b_3d_2 & a_1c_2d_3 & b_1c_1d_1\\
    a_2b_3c_1 & a_2b_1d_3 & a_2c_3d_1 & b_2c_2d_2\\
    a_3b_1c_2 & a_3b_2d_1 & a_3c_1d_2 & b_3c_3d_3.\\    
  \end{array}
  \]

  \noindent
  Each graph specifies all incidences that involve $A$, so we have to
  complete the list by those between $B,C$ and $D$, and look for
  isomorphisms. Remark that these incidences must occur outside $A$,
  to avoid creating a colorful incidence. Note the following
  dichotomy:

  \begin{itemize}
  \item If the lines of $A$ are concurrent and not coplanar, then two
    lines in $B \cup C \cup D$ that intersect different pairs of lines
    of $A$ cannot intersect outside of $A$.

  \item If the lines of $A$ are pairwise skew, then two lines in $B
    \cup C \cup D$ that intersect the same pair of lines of $A$ cannot
    intersect outside of $A$.
  \end{itemize}

  \noindent
  Consider the graph on the left. The incidences involving $A$ readily
  match $(I)$ and $(II)$.  If the lines of $A$ are concurrent and not
  coplanar, then we must have $b_1c_1d_1$, $b_2c_2d_2$, $b_3c_3d_3$,
  and therefore $(II)$. If the lines of $A$ are pairwise skew, then
  the only possibility avoiding a colorful incidence is to have
  $b_1c_3d_2$, $b_2c_1d_3$ and $b_3c_2d_1$, and
  therefore~$(I)$. Consider now the graph on the right. If the lines
  of $A$ are concurrent and not coplanar, then the $3$-consistency
  require that a line from $D$ passes through $b_3 \cap c_2$, which
  makes that incidence colorful since this point is already on
  $a_1$. Hence, the lines of $A$ can only be pairwise skew, and the
  only choice avoiding a colorful incidence is $b_1c_2d_3$,
  $b_2c_3d_1$, $b_3c_1d_2$; Applying the cyclic permutation $a \to d \to c \to b \to a$ gives $(II)$.

  \begin{remark}
    The set of incidences $(I)$ can be obtained from the monomials
    with positive sign in
    \[ \det \left[ \begin{array}{cccc} a_1 & b_1 & c_1 & d_1 \\ a_2 & b_2 & c_2 & d_2 \\ a_3 & b_3 & c_3 & d_3 \\ 1 & 1 & 1& 1\end{array} \right].\]
  \end{remark}

\subsubsection{Geometric realization}

\begin{lemma}
  The incidence structures $(I)$ and $(II)$ can be realized
  geometrically and their only non-planar geometric realizations are,
  respectively, a Reye-type configuration and a Desargues-type
  configuration.
\end{lemma}
\begin{proof}
  The fact that the two incidence structures are realizable can be
  seen directly from Figure~\ref{f:3333}. The second part of
  statement are essentially two ``geometric theorems'': it means the
  $12$ incidences of triples of lines from $(I)$ and $(II)$
  automatically guarantee that certain other collinearities or
  incidences always hold.

  Let us consider first the incidence structure $(II)$. Remark that
  the lines $b_1$ and $c_1$ intersect outside of a line of $A$ (since
  their intersection is on $d_1$) and intersect the same pair of lines
  of $A$, namely $\{a_2,a_3\}$. This implies that in any non-planar
  realization, the lines of $A$ must be concurrent and not
  coplanar. Moreover, in each of the triples $\{b_2,c_3,d_1\}$,
  $\{b_3,c_1,d_2\}$ and $\{b_1,c_2,d_3\}$, the lines intersect
  pairwise (we see that from their incidences) but are not concurrent
  (because any pairwise intersection occurs on a line of $A$); each of
  these triples is therefore coplanar. All three planes must contain
  the points of incidence of $b_1c_1d_1$, $b_2c_2d_2$ and
  $b_3c_3d_3$. As $A$ is already concurrent, all other color classes
  consist of pairwise skew lines, and these points are distinct. It
  follows that the three planes intersect in a line, and we have a
  Desargues-type configuration.

  Now let us consider $(I)$. We first observe that eight well-chosen
  incidence points are sufficient to determine all lines. For example,
  we may take the eight incidences to be
  \begin{equation}
    \label{eq:conf1}
    \begin{array}{l}
      a_{2} b_{3} c_{1} \\
      a_{3} b_{1} c_{2} \\
      a_{1} b_{3} d_{2} \\
      a_{3} b_{2} d_{1} \\
    \end{array}\qquad
    \begin{array}{ll}
      a_{1} c_{2} d_{3} \\
      a_{2} c_{3} d_{1} \\
      b_{1} c_{3} d_{2} \\
      b_{2} c_{1} d_{3} \\
    \end{array}
  \end{equation}
  Fixing these points arbitrarily in $\R^3$ yields a configuration of
  lines that satisfies $8$ out of the $12$ required incidences.  We
  now assume that three of the remaining four incidence points ({\em
    e.g.}, $a_{1} b_{2} c_{3}$, $a_{2} b_{1} d_{3}$, $a_{3} c_{1}
  d_{2}$) are at infinity, and correspond to orthogonal directions.
  We can assume this because any incidence properties of a geometric
  realization is preserved by projective transformations.  Enforcing
  the corresponding incidences (or collinearity of incidence points)
  forces the eight points in~\eqref{eq:conf1} to be the vertices of a
  parallelepiped.  This fixes a unique projective configuration, and
  there is no more freedom to impose the convergence of the remaining
  triple of lines ($b_{3} c_{2} d_{1}$).  On the other hand, this last
  incidence is always automatically satisfied: this geometric result
  is equivalent to the fact the diagonals of a parallelepiped meet at
  a single point. All of these alignments guarantee that the $12$
  incidence points are vertices in a Reye configuration.
\end{proof}

\begin{remark} The problem of determining whether or not a combinatorial
incidence structure admits a ``coordinatization'' can be addressed
algorithmically, by enforcing appropriate algebraic constraints on coordinates
using a computer algebra system (see, {\em e.g.},
~\cite{sturmfels1991computational}). We discovered the Reye realization
of the incidence structure $(I)$ using this method. This approach
also yields a computational proof that there is in fact a unique
projective realization.
\end{remark}

\subsection{Proof of Theorem~\ref{t:small-ii}}
\label{a:flat}

We now prove Theorem~\ref{t:small-ii}, which states that any
$3$-consistent colored set of lines $\Li = \Li_1 \cup \Li_2 \cup \Li_3
\cup \Li_4$ in $\R^3$ with no colorful incidence and concurrent colors
has at least $24$ lines or is contained in a $2$-plane. Let us start
with a decomposition lemma: 

\begin{lemma}\label{l:decomp}
 Let $\Li = \Li_1 \cup \Li_2 \cup \Li_3 \cup \Li_4$ be a
 $3$-consistent colored set of lines in $\R^3$ with no colorful
 incidence and concurrent color classes. If two color classes, each
 concurrent, are contained in two planes, then $\Li$ is contained in
 these planes and the subset of $\Li$ contained in each plane is
 $3$-consistent.
\end{lemma}
\begin{proof}
  Consider two color classes $\Li_1$ and $\Li_2$ contained in two
  planes $\Pi_1$ and $\Pi_2$. If all lines of $\Li_1$ are coplanar,
  then any line of $\Li \setminus \Li_1$ must intersect that plane
  into at least two distinct points, and the configuration is fully
  planar. Thus, each of $\Pi_1$ and $\Pi_2$ contains some lines of
  $\Li_1$ and $\Li_2$.

  Every line $\ell \in \Li\setminus(\Li_1 \cup \Li_2)$ must be
  contained in $\Pi_1$ or $\Pi_2$. To see this, assume otherwise. Then
  $\ell$ intersects the two planes in at most two points. These are
  the only points in which $\ell$ can intersect the lines of $\Li_1
  \cup \Li_2$. It must then be that both of these points are on a line
  of $\Li_1$ and on a line of $\Li_2$. Thus, any concurrence of $\ell$
  and a line of $\Li_1$ lies on a line of $\Li_2$; we cannot have
  $3$-consistency without a colorful incidence.

 Now, observe that $\Pi_1 \cap \Pi_2 = c_1c_2$. This implies that no
 $3$-incidence can occur on that line, as this would force a line from
 $\Li_1$ to contain $c_2$ (or vice versa) and $c_2$ would be a
 $4$-incidence. Thus, each of the subsets of $\Li$ contained in
 $\Pi_1$ and $\Pi_2$ is already $3$-consistent.
\end{proof}

We also use the following strengthening of Lemma~\ref{p:verysmall}:

\begin{lemma}\label{l:class2}
  Let $\Li = \Li_1 \cup \Li_2 \cup \Li_3 \cup \Li_4$ be a
  $3$-consistent colored set of lines in $\R^d$ with $|\Li_1| = 2$. If
  every color class is concurrent and $\Li$ contains no colorful
  incidence, then $|\Li_2| = |\Li_3| = |\Li_4|$ is even and at least
  $4$.
\end{lemma}

\noindent
Before we prove Lemma~\ref{l:class2}, let us see how it implies
Theorem~\ref{t:small-ii}.

\begin{proof}[Proof of Theorem~\ref{t:small-ii}]
  First, observe that the $3$-consistency and absence of colorful
  incidence implies the following fact: every line intersects every
  other color in at least two distinct points. Since $|\Li|<24$ there
  exists a color class of size at most $5$, say $\Li_1$.

  Let $c_2$ denote the center of concurrence of $\Li_2$. We examine
  the planes spanned by the lines of $\Li_1$ with $c_2$. Any such
  plane must contain two lines of $\Li_1$. Indeed, any line $\ell_1
  \in \Li_1$ intersects at least one line in $\Li_2$, which in turn
  intersects at least another line $\ell_1'$ in $\Li_1$. The planes
  spanned by the lines of $\Li_1$ with $c_2$ thus coincide by pairs,
  so there are at most $2$ such planes. Moreover, every line of
  $\Li_2$ intersects some line of $\Li_1$, and is therefore contained
  in one of these planes. By Lemma~\ref{l:decomp}, $\Li$ decomposes
  into two disjoint subsets, each of which is also $3$-consistent,
  without $4$-incidence, and with concurrent color classes. Now,
  observe that each subset must have at least $12$ lines: this is
  immediate for a subset with all color classes of size $3$ or more,
  and follows from Lemma~\ref{l:class2} for subsets with a color class
  of size~$2$.
\end{proof}

\subsection{Proof of Lemma~\ref{l:class2}}
\label{s:class2}

To analyze the situations where a color class has size two, it is
convenient to switch to a dual setting. Given a colored point set $P =
P_1 \cup P_2 \cup P_3 \cup P_4$, let a \emph{colorful alignment} in
$P$ is a line containing points of $P$ of all colors. We say that $P$
is \emph{$3$-consistent} if for any point $x \in P$ and any choice of
$2$ other colors, there is a line through $x$ that contains a point of
each of these colors. Now, a $4$-colored set of lines $\Li$ in $\R^d$
is $3$-consistent if and only if projecting $\Li$ to a generic
$2$-plane and taking the dual of these projected lines yields a
$3$-consistent colored point set. Similarly, colorful incidences are
mapped, by generic projection and duality, to colorful alignment.

\begin{proof}[Proof of Lemma~\ref{l:class2}]
  We prove the statement in its dual formulation. So, let $P = P_1
  \cup P_2 \cup P_3 \cup P_4$ be a colored planar point set that is
  $3$-consistent, has no colorful incidence and where each color class
  $P_i$ is contained in a line, denoted $\ell_i$. The $3$-consistency
  and lack of colorful alignment in $P$ implies that the $P_i$ are
  pairwise disjoint and that for any $i \neq j$ we have $\ell_i \cap
  P_j = \emptyset$.

  Let us write $P_1 = \{v,h\}$ and put $P' = P_2 \cup P_3 \cup
  P_4$. We define a graph $G$ with vertex set $P'$ and an edge in $G$
  between $p$ and $q$ if they have different colors and are aligned
  with a point of $P_1$. We orient the graph $G$ by orienting the
  edges from $P_2$ to $P_3$, from $P_3$ to $P_4$, and from $P_4$ to
  $P_2$. Every vertex has exactly two edges, one incoming and one
  outgoing; this follows from the $3$-consistency and the fact that
  $\ell_i \cap P_j = \emptyset$ for $i \neq j$. Thus, $G$ is a
  disjoint union of cycles.

  To every edge $(p,q)$ in $G$ we associate two parameters: the pair
  $(i,j)$ of color classes of $p$ and $q$, respectively, and the
  vertex of $P_1$ that is collinear with $p$ and $q$. For $(i,j) \in
  \{(2,3), (3,4), (4,2)\}$ and $c \in P_1$ we let $p_{i,j}^c$ denote
  the projection from $\ell_i$ to $\ell_j$ with center~$c$. Let $T$
  denote the set of vertices in $P_3$ that are aligned with their
  successors and the point $h$ from $P_1$.  Along a cycle, $(i,j)$ is
  $3$-periodic, by the choice of orientation of the edges, and $x$ is
  $2$-periodic, by $3$-consistency and lack of $4$-alignment. Thus,
  every cycle in $G$ has length some multiple of $6$ and has a vertex
  in $T$. Moreover, this vertex must be a fixed point of some power of
  the function $f$ defined by:
  \[ f:\left\{\begin{array}{rcl} \ell_3 & \to & \ell_3\\
  (x,y) & \mapsto & p_{2,3}^v \circ p_{4,2}^h \circ p_{3,4}^v \circ
  p_{2,3}^h \circ p_{4,2}^v \circ p_{3,4}^h(x,y)
  \end{array}\right.\]

  To compute $f$ it is convenient to apply a projective transform that
  maps $v$ and $h$ to infinity, to respectively the vertical and
  horizontal directions. This preserves incidences in $P$, being
  understood that a line contains $v$ (resp. $h$) if and only if it is
  vertical (resp. horizontal). In particular, none of $\ell_2$,
  $\ell_3$ or $\ell_4$ is horizontal or vertical, so we can
  parametrize $\ell_i$ by $(x, \alpha_i x + \beta_i)$. We then have:
  \[ p_{i,j}^v(x,y) = (x,\alpha_j x + \beta_j) \quad \hbox{and} \quad  p_{i,j}^h(x,y) = \pth{\frac{y-\beta_j}{\alpha_j},y}.\]
  Let us put the origin of our frame at $\ell_2 \cap \ell_3$ so that
  $\beta_2 = \beta_3 = 0$. An elementary computation then yield that
  the point $f\pth{x,\alpha_3 x+\beta_3}$ has $x$-coordinate $x+c$
  where
  \[ c = \frac{\alpha_3-\alpha_2}{\alpha_3\alpha_2}\beta_4.\]
  If $c \neq 0$ then no power of $f$ has a fixed point, so we must
  have $c=0$. Since $\alpha_2 \neq \alpha_3$, we must have $\beta_4=0$
  and $\ell_4$ goes through $\ell_2 \cap \ell_3$.

  Altogether, we proved that the lines supporting $P_2$, $P_3$ and
  $P_4$ are concurrent, and that $P'$ decomposes into some number $r$
  of $6$-cycles $\gamma_1, \ldots, \gamma_r$ of the form
  \[ \begin{aligned}
    a \in \ell_3 \to b = p_{3,4}^h(a) \to c =  p_{4,2}^v(b) & \to d = p_{2,3}^h(c)\\
    & \to e= p_{3,4}^v(d) \to f = p_{4,2}^h(e) \to a = p_{2,3}^v(f).
    \end{aligned}\]
  In particular, $|P_2|=|P_3|=|P_4| = 2r$. Lemma~\ref{p:verysmall}
  already established that we must have $r>1$.
\end{proof}

Interestingly, for every $r>1$, there exists a $3$-consistent
configuration with $|P_1|=2$ and $|P_2|=|P_3|=|P_4| = 2r$ that has no
colorful alignment.

\bibliography{ref}

\end{document}